\documentclass[letterpaper,10pt,conference]{ieeeconf}
\IEEEoverridecommandlockouts
\usepackage{cite}
\usepackage{stmaryrd}
\usepackage{amsmath,amssymb,amsfonts}
\usepackage{algorithm}
\usepackage{algpseudocode}
\usepackage{makecell}
\usepackage{tabularx}

\floatname{algorithm}{Protocol}

\algrenewcommand{\algorithmicrequire}{\textbf{Initialization:}}
\algrenewcommand{\algorithmicensure}{\textbf{Output:}}

\algrenewcommand{\algorithmicindent}{1.2em}

\newtheorem{theorem}{Theorem}
\newtheorem{assumption}{Assumption}

\newcommand{\Enc}{\mathrm{Enc}}
\newcommand{\Dec}{\mathrm{Dec}}

\usepackage{braket}
\usepackage{graphicx}
\usepackage{textcomp}
\usepackage{xcolor}
\def\BibTeX{{\rm B\kern-.05em{\sc i\kern-.025em b}\kern-.08em
    T\kern-.1667em\lower.7ex\hbox{E}\kern-.125emX}}
\begin{document}

\title{Explicit Model Predictive Control with Quantum Encryption
}

\author{Yingjie Mi, Zihao Ren, Lei Wang, Daniel E. Quevedo,  and Guodong Shi
\thanks{Y. Mi, Z. Ren and G. Shi are with the Australian Centre for Robotics, School of Aerospace, Mechanical and Mechatronic Engineering, The University of Sydney, Sydney, NSW, Australia. Emails: yingjie.mi@sydney.edu.au, zihao.ren@sydney.edu.au guodong.shi@sydney.edu.au}
\thanks{L. Wang is with the State Key Laboratory of Industrial Control Technology, Institute of Cyber-Systems and Control, Zhejiang University, Zhejiang, China. E-mail: lei.wangzju@zju.edu.cn.}
\thanks{D. Quevedo is with the School of Electrical and Computer Engineering, The University of Sydney, Sydney, NSW, Australia. Email: daniel.quevedo@sydney.edu.au}
}

\maketitle

\begin{abstract}
This paper studies quantum-encrypted explicit MPC for constrained discrete-time linear systems in a cloud-based architecture. A finite-horizon quadratic MPC problem is solved offline to obtain a piecewise-affine controller. Shared quantum keys generated from Bell pairs and protected by quantum key distribution are used to encrypt the online control evaluation between the sensor and actuator. Based on this architecture, we develop a lightweight encrypted explicit MPC protocol, prove exact recovery of the plaintext control action, and characterize its computational efficiency. Numerical results demonstrate lower online complexity than classical encrypted MPC, while security is discussed in terms of confidentiality of plant data and control inputs.
\end{abstract}


\section{Introduction}
Networked and cloud-based control systems coordinate sensing, computation, communication, and actuation over shared communication networks. They arise in applications such as smart grids, robotics, building automation, and intelligent transportation, where they improve scalability and flexibility. However, communication also exposes plant states, control inputs, and controller-side computations to eavesdropping, inference, and malicious manipulation, raising security and privacy concerns \cite{6257525,6545301}.

Encrypted control addresses these risks by enabling control laws to be evaluated on protected data, thereby preserving confidentiality during both communication and computation \cite{9438025,SCHLUTER2023100913,7403296}. Existing approaches based on homomorphic encryption, secret sharing, and secure multi-party computation have been developed for linear feedback, dynamic controllers, and model predictive control (MPC) \cite{FAROKHI201713,9029342,alexandru2019secure,7403296,9086151,8126799,8619835,schluter2020encrypted}. Among these, MPC is especially attractive because it explicitly handles constraints while optimizing closed-loop performance.

Despite this progress, encrypted MPC still faces two main challenges. First, many implementations rely on costly ciphertext arithmetic, leading to significant online overhead \cite{8126799}. Second, many representative schemes are built on classical public-key cryptography and may therefore be vulnerable to future quantum attacks \cite{chen2016report}.

Motivated by these issues, this paper studies encrypted MPC for constrained discrete-time linear systems using quantum-generated keys. We consider a finite-horizon quadratic MPC problem in explicit form, so that the online controller is a piecewise affine law computed offline. To preserve confidentiality in a cloud-based architecture, a quantum channel is established between the sensor and actuator to generate shared random keys from entangled Bell pairs, which can be further protected by quantum key distribution \cite{ekert1991quantum,nielsen2010quantum}. These keys are available only to the sensor and actuator, while the cloud has access neither to the keys nor to plaintext state and input signals.

As illustrated in Fig.~\ref{fig:system_diagram}, we develop a lightweight quantum-encrypted explicit MPC protocol that extends the quantum-encrypted state-feedback control framework in \cite{ren2025quantum} to constrained predictive control with a region-dependent piecewise affine structure \cite{bemporad2002explicit}. The sensor performs local region identification and encryption, the cloud evaluates the encrypted affine law, and the actuator decrypts the control input. We establish exact recovery of the plaintext explicit MPC law and lower online computational complexity than representative classical encrypted MPC methods under matched control accuracy.

\begin{figure}
    \centering
    \includegraphics[width=0.8\linewidth]{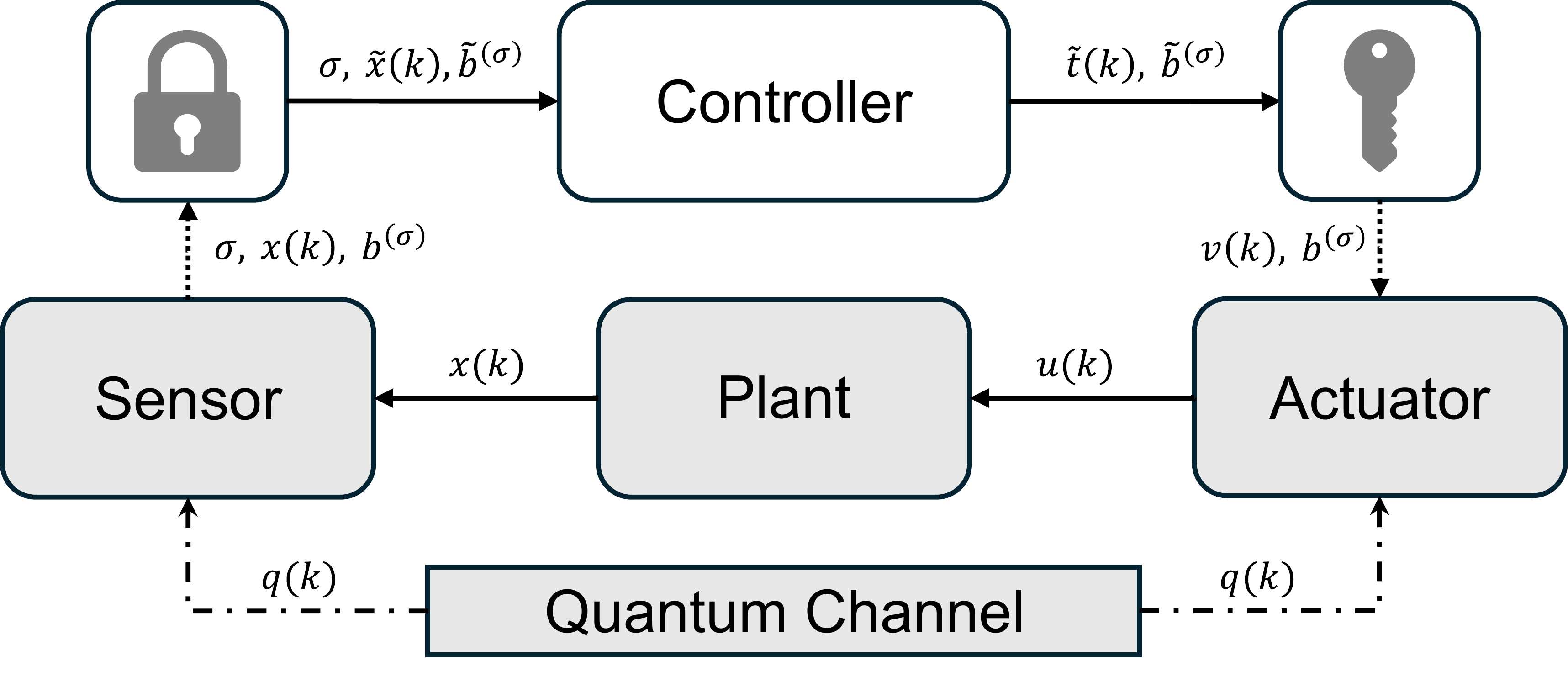}
    \caption{Architecture of the proposed quantum-encrypted explicit MPC. The sensor identifies the active region and encrypts the state and affine offset, the cloud evaluates the encrypted affine law, and the actuator decrypts the control input using shared keys generated via a quantum channel.}
    \label{fig:system_diagram}
\end{figure}

The remainder of this paper is organized as follows. Section~II formulates the explicit MPC problem and the quantum-key-based encrypted control architecture. Section~III presents the proposed quantum-encrypted explicit MPC protocol and its main theoretical properties. Section~IV reports numerical results and comparisons with classical encrypted MPC baselines. Section~V concludes the paper.

\section{Background}
\subsection{Linear-Quadratic MPC}
We consider a discrete-time linear time-invariant (LTI) system
\begin{equation}
x(k+1)=Ax(k)+Bu(k), \qquad y(k)=Cx(k),
\label{eq:lti}
\end{equation}
where $x(k)\in\mathbb{R}^{n}$, $u(k)\in\mathbb{R}^{m}$, and $y(k)\in\mathbb{R}^{p}$ denote the state, input, and measured output, respectively, and $k\in\mathbb{N}$ is the discrete-time index.
The matrices $A\in\mathbb{R}^{n\times n}$, $B\in\mathbb{R}^{n\times m}$, and $C\in\mathbb{R}^{p\times n}$ are assumed known.
The state and input satisfy $x(k)\in\mathcal{X}$ and $u(k)\in\mathcal{U}$,
where $\mathcal{X}\subset\mathbb{R}^{n}$ and $\mathcal{U}\subset\mathbb{R}^{m}$
are compact convex polyhedra.
 
\label{subsec:ocp}
At each time $k$, given the current state $x(k)$, we solve the following constrained finite-horizon linear-quadratic optimal control problem
\begin{equation}
\begin{aligned}
\min_{\substack{\bar x(0),\ldots,\bar x(N)\\ \bar u(0),\ldots,\bar u(N-1)}} \quad &
\|\bar x(N)\|_{P}^{2} +
\sum_{\kappa=0}^{N-1}\bigl(\|\bar x(\kappa)\|_{Q}^{2}+\|\bar u(\kappa)\|_{R}^{2}\bigr)\\
\text{s.t.}\quad 
\bar x(0) & = x(k),\\
\bar x(\kappa+1) & = A \bar x(\kappa)+B \bar u(\kappa), \qquad \kappa=0,\dots,N-1,\\
\bar x(\kappa) & \in\mathcal{X},\; \bar u(\kappa)\in\mathcal{U}, \qquad \kappa=0,\dots,N-1,\\
\bar x(N) & \in\mathcal{T}.
\end{aligned}
\label{eq:ocp}
\end{equation}
where $N\in\mathbb{N}$ is the prediction horizon, $Q\succeq 0$, $R\succ 0$, and $P\succeq 0$ are weighting matrices, and $\mathcal{T}\subseteq\mathcal{X}$ is a terminal set.
The receding-horizon control law applies only the first input,
\begin{equation}
u(k)=\bar u^\star(0),
\label{eq:receding_horizon}
\end{equation}
where $\bar u^\star(0)$ denotes the first element of the optimal input sequence.

\subsection{Explicit MPC}
Under the standing assumptions that $\mathcal{X}$, $\mathcal{U}$, and $\mathcal{T}$ are convex polyhedra and $R\succ 0$, problem \eqref{eq:ocp} is a strictly convex QP with a unique optimizer, which can be equivalently expressed in condensed form (see, e.g., \cite{bemporad2002explicit})
\begin{equation}
\begin{aligned}
z^\star(x)=\arg\min_{z\in\mathbb{R}^{mN}} \quad &
\frac{1}{2}z^\top H z + x^\top F^\top z \\
\text{s.t.}\quad & G z \le E x + h .
\end{aligned}
\label{eq:qp_condensed}
\end{equation}
Here, the decision variable stacks the predicted inputs,
\begin{equation}
z := \begin{bmatrix}\tilde u(0)^\top & \cdots & \tilde u(N-1)^\top\end{bmatrix}^\top \in \mathbb{R}^{mN},
\label{eq:z_def}
\end{equation}
and $H\in\mathbb{R}^{mN\times mN}$ is symmetric positive definite.
All matrices $(H,F,G,E,h)$ are constant and can be precomputed offline, while the parameter is the current state $x=x(k)$.
The MPC input is obtained from $z^\star(x)$ via
\begin{equation}
u(k)=C z^\star(x(k)),\; C:=\begin{bmatrix} I_m & 0 & \cdots & 0 \end{bmatrix}\in\mathbb{R}^{m\times mN}.
\label{eq:u_from_z}
\end{equation}

In \eqref{eq:qp_condensed}, the current state $x=x(k)$ is considered as a parameter through affine terms in both the objective and the constraints, hence \eqref{eq:qp_condensed} is a strictly convex multi-parametric QP (mpQP) in $x$. The following is a classical result on the solutions to the MPC as a piecewise affine feedback control law established in \cite{bemporad2002explicit}. 

\medskip

\begin{theorem}
Under standard nondegeneracy conditions, there exists a finite polyhedral partition $\{\mathcal{P}^{(\sigma)}\}_{\sigma=1}^{s}$ of the feasible state set such that the explicit MPC feedback is piecewise affine
\begin{equation}
u(k)=K^{(\sigma)}x(k)+b^{(\sigma)} \quad \text{if } x(k)\in\mathcal{P}^{(\sigma)},
\label{eq:pwa}
\end{equation}
where $\sigma\in\{1,\ldots,s\}$ denotes the region index, $s\in\mathbb{N}$ is the
number of regions, and $K^{(\sigma)}\in\mathbb{R}^{m\times n}$,
$b^{(\sigma)}\in\mathbb{R}^{m}$.
\label{thm:explicit_mpc_pwa}
\end{theorem}

\subsection{Explicit MPC with Encryption for Networked Systems}
Privacy concerns in networked control motivate encrypted schemes that hide measurements and actuation signals from untrusted servers and eavesdroppers. Here we briefly review a homomorphic-encryption realization of explicit MPC in a cloud-assisted architecture, which serves as the classical baseline for the proposed QE-MPC. Throughout, $\tilde{\cdot}$ denotes ciphertexts, and $\Enc_{\mathrm{HE}}(\cdot)$ and $\Dec_{\mathrm{HE}}(\cdot)$ denote the corresponding encryption and decryption mappings.
Additively homomorphic schemes such as Paillier~\cite{10.1007/3-540-48910-X_16} support ciphertext-domain addition and multiplication by a known plaintext constant~\cite{9438025}. For plaintext integers $z_1,z_2$ and scalar $a$, it hold that
\begin{equation}
\tilde{z}_1 \oplus \tilde{z}_2 = \Enc_{\mathrm{HE}}(z_1+z_2), 
\qquad
a \odot \tilde{z} = \Enc_{\mathrm{HE}}(a z).
\label{eq:he_ops}
\end{equation}
Hence affine maps can be evaluated directly in the encrypted domain. Since the plaintext space is typically integer-valued, real-valued control laws are implemented through fixed-point encoding, which introduces quantization effects~\cite{9029342}.

Because constrained MPC involves operations such as comparison and projection that are not natively supported by additive HE~\cite{9438025}, a common approach is to solve the constrained optimization offline and encrypt only the online evaluation of the resulting explicit MPC law~\cite{9438025,schluter2020encrypted}. At time $k$, the sensor measures $x(k)$ and identifies the active region
\begin{equation}
\sigma \gets \pi\big(x(k)\big)\quad \text{s.t.}\quad x(k)\in \mathcal P^{(\sigma)}.
\label{eq:he_region_id}
\end{equation}
It then sends $\textsf{Msg}^{\mathrm{HE}}_{S\to C}(k):=(\sigma,\tilde x(k))$, where $\tilde x(k)=\Enc_{\mathrm{HE}}(x(k))$. Using $\sigma$, the cloud selects $(K^{(\sigma)},b^{(\sigma)})$ and computes
\begin{equation}
\tilde u(k)=\Big(K^{(\sigma)} \odot \tilde x(k)\Big)\oplus \tilde b^{(\sigma)},
\;
\tilde b^{(\sigma)}=\Enc_{\mathrm{HE}}\!\big(b^{(\sigma)}\big),
\label{eq:he_affine_eval}
\end{equation}
after which the actuator decrypts and applies
\begin{equation}
u(k)=\Dec_{\mathrm{HE}}\!\big(\tilde u(k)\big).
\label{eq:he_dec_apply}
\end{equation}

This architecture prevents direct recovery of plaintext states and inputs by passive eavesdroppers and by an honest-but-curious cloud without decryption keys. However, it still has three main limitations:
\begin{enumerate}
    \item the cloud observes the active region index $\sigma$;
    \item fixed-point realization introduces quantization overhead; and
    \item classical HE relies on costly public-key arithmetic and computational hardness assumptions that are not generally regarded as quantum resilient.
\end{enumerate}
These limitations motivate the quantum encrypted explicit MPC architecture developed next.

\section{Quantum Encrypted MPC}
We now describe the proposed quantum-encrypted explicit MPC architecture, see Fig.~\ref{fig:system_diagram}. The framework consists of a sensor, an untrusted cloud controller, and an actuator. The sensor performs local region identification and encryption, the cloud evaluates the encrypted affine law using the active gain in plaintext, and the actuator reconstructs the control input using shared keys established with the sensor. Thus, the cloud has access neither to the shared keys nor to plaintext state and input signals.

\subsection{Quantum Encrypted Explicit MPC}
The proposed architecture consists of a sensor (S), an untrusted cloud controller (C), and an actuator (A). The sensor and actuator are connected by a quantum channel for key establishment, while the sensor--cloud and cloud--actuator links are classical channels for control-related communication.

At each discrete-time instant $k\in\mathbb{N}$, the quantum channel distributes $w_q$ Bell pairs between S and A. For $\ell=1,\dots,w_q$, the $\ell$th pair is prepared as
\begin{equation}
\ket{\psi^\ast}^{(\ell)}=\frac{1}{\sqrt{2}}\big(\ket{0}_S\ket{0}_A+\ket{1}_S\ket{1}_A\big),
\label{eq:bell_state}
\end{equation}
with S holding the first qubit and A the second. By measuring their qubits in the computational basis, S and A obtain an identical random key string $q(k)\in\{0,1\}_{w_q}$, while the cloud has no access to $q(k)$, see \cite{nielsen2010quantum}.

We instantiate the quantum-encrypted primitives through an exponential--logarithmic realization introduced in \cite{ren2025quantum}. Throughout, $\tilde{\cdot}$ denotes ciphertext-domain quantities and $[\cdot]_i$ the $i$th component of a vector. At each time $k\in\mathbb{N}$, the sensor and actuator share a fresh key stream $q(k)\in\{0,1\}_{w_q}$ generated by the quantum channel. Following \cite{ren2025quantum}, let $d=n+m$ and partition $q(k)$ into $d$ groups of $w_b$ bits, with $dw_b=w_q$, i.e.,
\begin{equation}
q(k)=\overline{b_{1,w_b-1}\cdots b_{1,0}\; b_{2,w_b-1}\cdots b_{2,0}\;\cdots\;
b_{d,w_b-1}\cdots b_{d,0}}(k),
\label{eq:qk_bits}
\end{equation}
where $b_{i,j}(k)\in\{0,1\}$ is the $j$th bit in the $i$th group. The first $n$ groups are assigned to the state components and the remaining $m$ groups to the affine offset components. For each $i=1,\dots,d$, define
\begin{equation}
\beta_i(k):=
-\big(2^{w_b-1}+1\big)b_{i,w_b-1}(k)
+\sum_{j=0}^{w_b-2}2^j b_{i,j}(k)+1,
\label{eq:beta_unified}
\end{equation}
so that $\beta_i(k)\neq 0$ for all $i$ and $k$.

For a scalar plaintext $z$ associated with coefficient $\beta(k)$, encryption and decryption are defined as
\begin{align}
\tilde z = f_{\mathrm{Enc}}(z,q(k))
&:= \exp\!\Big(\frac{z}{\beta(k)}\Big)
\label{eq:fEnc_explog}\\
f_{\mathrm{Dec}}(\tilde z,q(k))
&:= \beta(k)\ln(\tilde z)
\label{eq:fDec_explog}
\end{align}
Accordingly, the sensor encrypts the measured state and region-dependent offset componentwise as
\begin{equation}
[\tilde x(k)]_i=\exp\!\Big(\frac{[x(k)]_i}{\beta_i(k)}\Big),\quad i=1,\dots,n,
\label{eq:enc_x}
\end{equation}
\begin{equation}
[\tilde b^{(\sigma)}]_j=\exp\!\Big(\frac{[b^{(\sigma)}]_j}{\beta_{n+j}(k)}\Big),\quad j=1,\dots,m.
\label{eq:enc_b}
\end{equation}

Given $K^{(\sigma)}$ and $\tilde x(k)$, the cloud evaluates the encrypted linear term through
\begin{equation}
\tilde t(k):=f_{\mathrm{Con}}\!\big(K^{(\sigma)},\tilde x(k)\big),
\label{eq:fCon_packed}
\end{equation}
where the $(j,i)$th entry is
\begin{equation}
[\tilde t(k)]_{j,i}=\big([\tilde x(k)]_i\big)^{[K^{(\sigma)}]_{j,i}}.
\label{eq:fCon_element}
\end{equation}
The actuator then reconstructs the linear part by
\begin{equation}
[v(k)]_j=\sum_{i=1}^{n}\beta_i(k)\ln\!\big([\tilde t(k)]_{j,i}\big),\qquad j=1,\dots,m.
\label{eq:dec_aggregate}
\end{equation}

In Protocol~\ref{prot:qe_empc_protocol}, the cloud holds the explicit MPC gain library and accesses $K^{(\sigma)}$ in plaintext, while the sensor holds the region-dependent offset $b^{(\sigma)}$ in plaintext for local encryption before transmission.

\begin{algorithm}[t]
\caption{Quantum Encrypted Explicit MPC}
\label{prot:qe_empc_protocol}
\begin{algorithmic}[1]
\Require $\{\mathcal P^{(\sigma)},K^{(\sigma)},b^{(\sigma)}\}_{\sigma=1}^{s}$, \;
$f_{\mathrm{Enc}},f_{\mathrm{Con}},f_{\mathrm{Dec}}$
\For{$k\in\mathbb{N}$}
  \State S and A share a fresh key stream $q(k)\in\{0,1\}_{w_q}$.
  \State S measures $x(k)$, sets $\sigma\gets\pi(x(k))$ with $x(k)\in\mathcal P^{(\sigma)}$, and computes
  $\tilde x(k)\gets f_{\mathrm{Enc}}(x(k),q(k))$,
  $\tilde b^{(\sigma)}\gets f_{\mathrm{Enc}}(b^{(\sigma)},q(k))$.
  \State S $\to$ C: $\big(\sigma,\tilde x(k),\tilde b^{(\sigma)}\big)$.
  \State C selects $K^{(\sigma)}$ and computes
  $\tilde t(k)\gets f_{\mathrm{Con}}\!\big(K^{(\sigma)},\tilde x(k)\big)$.
  \State C $\to$ A: $\big(\tilde t(k),\tilde b^{(\sigma)}\big)$.
  \State A computes
  $v(k)\gets f_{\mathrm{Dec}}(\tilde t(k),q(k))$,
  $b^{(\sigma)}\gets f_{\mathrm{Dec}}(\tilde b^{(\sigma)},q(k))$,
  and applies $u(k)\gets v(k)+b^{(\sigma)}$.
\EndFor
\end{algorithmic}
\end{algorithm}

Compared with the classical encrypted explicit MPC scheme in \cite{8126799}, the proposed method adopts the same cloud-assisted architecture for constrained linear systems, with offline explicit MPC, local region identification, cloud-side evaluation of the active affine segment, and actuator-side input reconstruction. Thus, both approaches preserve the same piecewise affine control structure. The main difference lies in the encrypted realization: \cite{8126799} uses Paillier homomorphic encryption with robust MPC compensation for fixed-point quantization, whereas QE-MPC uses a quantum-key-enabled exponential--logarithmic realization with shared keys generated from Bell pairs and also admits a separate quantized implementation. Consequently, QE-MPC preserves the original explicit MPC law while reducing online complexity.

\subsection{Effectiveness}
Here, \emph{effectiveness} means that the encrypted protocol reproduces the original explicit MPC law exactly: after encryption, cloud-side evaluation, and decryption, the recovered input is identical to the corresponding plaintext control input.
\begin{theorem}
\label{thm:effectiveness}
Let $x(k)$ be the measured state and let $\sigma=\pi(x(k))$ be the region index identified at the sensor such that
$x(k)\in\mathcal P^{(\sigma)}$.
Then the control input $u(k)$ generated by ~\eqref{eq:fEnc_explog}--\eqref{eq:fDec_explog} satisfies
\begin{equation}
u(k)=K^{(\sigma)}x(k)+b^{(\sigma)},\qquad \forall k\in\mathbb N,
\label{eq:correctness}
\end{equation}
where $(K^{(\sigma)},b^{(\sigma)})$ are the PWA gains defined in \eqref{eq:pwa}.
\end{theorem}

\begin{proof}
For all $k\in\mathbb N$ and let $\sigma=\pi(x(k))$ such that $x(k)\in\mathcal P^{(\sigma)}$.
By \eqref{eq:fDec_explog} and \eqref{eq:fEnc_explog}, for each $j=1,\dots,m$ we have
\[
\begin{aligned}
u_j(k)
&= \beta_{b,j}(k)\ln\!\big([\tilde b^{(\sigma)}]_j\big)
  +\sum_{i=1}^n \beta_i(k)\ln\!\big([\tilde t(k)]_{j,i}\big)\\
&= \beta_{b,j}(k)\ln\!\big([\tilde b^{(\sigma)}]_j\big)
  +\sum_{i=1}^n \beta_i(k)\ln\!\Big(\big([\tilde x(k)]_i\big)^{[K^{(\sigma)}]_{j,i}}\Big)\\
&= \beta_{b,j}(k)\ln\!\big([\tilde b^{(\sigma)}]_j\big)
  +\sum_{i=1}^n \beta_i(k)\,[K^{(\sigma)}]_{j,i}\,\ln\!\big([\tilde x(k)]_i\big),
\end{aligned}
\]
Then from \ref{eq:fEnc_explog} we obtain the cancellation identities
\[
\begin{aligned}
\ln\!\big([\tilde x(k)]_i\big)=\frac{[x(k)]_i}{\beta_i(k)},\quad i=1,\dots,n,\\
\ln\!\big([\tilde b^{(\sigma)}]_j\big)=\frac{[b^{(\sigma)}]_j}{\beta_{b,j}(k)},\quad j=1,\dots,m.
\end{aligned}
\]
Substituting these identities yields
\[
\begin{aligned}
u_j(k)
&= \beta_{b,j}(k)\frac{[b^{(\sigma)}]_j}{\beta_{b,j}(k)}
  +\sum_{i=1}^n \beta_i(k)\,[K^{(\sigma)}]_{j,i}\,\frac{[x(k)]_i}{\beta_i(k)}\\
&= [b^{(\sigma)}]_j+\sum_{i=1}^n [K^{(\sigma)}]_{j,i}\,[x(k)]_i,\qquad j=1,\dots,m.
\end{aligned}
\]
Stacking $j=1,\dots,m$ gives \eqref{eq:correctness}.
\end{proof}

\subsection{Quantization and Computational Efficiency}
We first introduce quantization models for Paillier-based MPC and QE-MPC, since finite-precision implementation is necessary in practice and directly affects computational cost. In Paillier-based MPC, quantization is required to encode real-valued signals into an integer plaintext space, whereas in QE-MPC it is used to characterize finite-precision implementation and enable a matched-accuracy complexity comparison.

For Paillier, we use the standard uniform fixed-point set parameterized by $\rho,\gamma,\delta\in\mathbb{N}$ with $\rho\ge 1$:
\begin{equation}
\mathbb{Q}_{\rho,\gamma,\delta}
:=
\left\{
-\rho^{\gamma}+k\rho^{-\delta}\ \big|\ k=0,1,\ldots,2\rho^{\gamma+\delta}-1
\right\},
\end{equation}
which has spacing $\rho^{-\delta}$. Let $g:\mathbb{R}\to\mathbb{Q}_{\rho,\gamma,\delta}$ satisfy
\begin{equation}
|g(x)-x|\le \rho^{-\delta},
\qquad \forall x\in[-\rho^{\gamma},\rho^{\gamma}],
\label{eq:q_err_darup}
\end{equation}
and define
\begin{equation}
q(x):=\rho^{\delta}g(x)\in\mathbb{Z},
\qquad
\mu(x):=q(x)\bmod n \in \{0,\ldots,n-1\},
\label{eq:pt_encoding}
\end{equation}
where $n$ is the Paillier modulus.

For QE-MPC, given $w\in\mathbb{N}$ and $v\in\mathbb{R}$, define
\begin{equation}
y:=g(v):=
\begin{cases}
v, & v>1,\\[0.3ex]
2-\frac{1}{v}, & v\le 1,
\end{cases}
\label{eq:qec_g}
\end{equation}
and let
\[
\xi_w:=\sum_{j=0}^{w-1}2^{-j}a_j,\qquad
\eta:=2^{w-1}y-\lfloor 2^{w-1}y\rfloor .
\]
Then $Q_w(v)=\overline{a_{w-1}a_{w-2}\cdots a_0}\in\{0,1\}^w$ is generated by the stochastic map $h_w$ such that
\begin{align}
\mathbb{P}\!\left(\xi_w=\frac{\lfloor 2^{w-1}y\rfloor}{2^{w-1}}\right)&=1-\eta, \nonumber\\
\mathbb{P}\!\left(\xi_w=\frac{\lfloor 2^{w-1}y\rfloor+1}{2^{w-1}}\right)&=\eta.
\label{eq:qempc_hw}
\end{align}
For $v\in[1/2,2]$, this quantizer is unbiased and satisfies
\begin{equation}
\mathbb{E}\!\left[\xi_w-y\right]=0,
\qquad
\mathbb{E}\!\left[(\xi_w-y)^2\right]\le 2^{-2w}.
\label{eq:qec_moments}
\end{equation}

\paragraph{Bit-cost model}
Let $L:=\mathrm{bitlen}(N)$, where Paillier ciphertext arithmetic is modulo $N^2$. Under the classical schoolbook bit model \cite{menezes2018handbook}, modular multiplication costs $O(L^2)$, and square-and-multiply exponentiation with a $b$-bit exponent costs $O(bL^2)$. With fixed-point encoding and $\mathrm{bitlen}(z)\le L$, Paillier $\Enc_{\mathrm{HE}}$ and $\Dec_{\mathrm{HE}}$ each cost $O(L^3)$ in the worst case, while $\oplus$ and $\odot$ cost $O(L^2)$ and $O(b_KL^2)$, respectively, where $b_K:=\max_{j,i}\mathrm{bitlen}(\hat K_{j,i})$.

\begin{assumption}
\label{ass:eq_align}
Fix an accuracy target $\varepsilon_q>0$.
Choose the Paillier fixed-point parameter $\delta$ and the QE-MPC quantization bit budget $w$ such that
\begin{equation}
\rho^{-\delta}\le \varepsilon_q,
\qquad
2^{-w}\le \varepsilon_q,
\qquad
p\ge w,
\label{eq:eq_align}
\end{equation}
where $p$ denotes the internal numerical precision (in bits) for ciphertext-domain arithmetic in Protocol~1.
\end{assumption}

\begin{assumption}
\label{ass:no_wrap}
Paillier plaintext parameters are selected such that no modular wrap-around occurs during one-step evaluation of
$u=K^{(\sigma)}x+b^{(\sigma)}$, so ciphertext-domain operations correspond to integer-domain operations after decoding.
\end{assumption}

\begin{theorem}
\label{thm:bit_ops_equal_accuracy}
Consider the evaluation of the explicit MPC law \eqref{eq:pwa} in a fixed region $\sigma$ implemented using the Paillier plaintext encoding \eqref{eq:pt_encoding} and QE-MPC Protocol~1.
Fix an accuracy target $\varepsilon_q>0$ and choose
\begin{equation}
\delta \ge \left\lceil \log_\rho\frac{1}{\varepsilon_q}\right\rceil,\qquad
w \ge \left\lceil \log_2\frac{1}{\varepsilon_q}\right\rceil,\qquad
p \ge w.
\label{eq:delta_w_choices}
\end{equation}
If Assumptions~\ref{ass:eq_align}--\ref{ass:no_wrap} hold, then the worst-case \textbf{total} per-cycle bit-complexity satisfies
\begin{equation}
\begin{aligned}
C_{\mathrm{tot}}^{\mathrm{HE}}
=& O\!\big((n+2m)L^3 + mn\,(b_K+1)\,L^2\big),
\\
C_{\mathrm{tot}}^{\mathrm{QE}}
=& O\!\big((mn+n+m)p^3\big),
\label{eq:total_bit_bounds}
\end{aligned}
\end{equation}
and the corresponding per-party bounds are given in Table~\ref{tab:bit_bounds_both}.
\end{theorem}

\begin{table}[htbp]
\centering
\caption{Worst-case per-cycle bit-complexity upper bounds per party}
\label{tab:bit_bounds_both}
\begin{tabular}{lcc}
\hline
Party & Paillier-EMPC & QE-MPC (Protocol 1) \\
\hline
Sensor 
& $O((n+m)L^3)$
& $O((n+m)p^3)$
\\
Controller 
& $O(mn(b_K+1)L^2)$
& $O(mn(p^3+p^2))$
\\
Actuator 
& $O(mL^3)$
& $O((mn+m)p^3 + mn\,p^2)$
\\
\hline
Total
& \makecell{$O((n+2m)L^3$ \\ $+\, mn(b_K+1)L^2)$}
& $O((mn+n+m)p^3)$
\\
\hline
\end{tabular}
\end{table}

\begin{table}[htbp]
\centering
\caption{Per-step primitive counts for encrypted explicit MPC evaluation
$u(k)=K^{(\sigma)}x(k)+b^{(\sigma)}$}
\label{tab:primitive_counts}
\begin{tabular}{l|c|c}
\hline
 & \textbf{Paillier-EMPC} & \textbf{QE-MPC (Protocol 1)} \\
\hline
Sensor (S)
&
\begin{tabular}[c]{@{}c@{}}
$\Enc_{\mathrm{HE}}$: $n+m$
\end{tabular}
&
\begin{tabular}[c]{@{}c@{}}
$f_{\mathrm{Enc}}$: $n+m$
\end{tabular}
\\
\hline
Controller (C)
&
\begin{tabular}[c]{@{}c@{}}
$\odot$: $mn$\\
$\oplus$: $mn$
\end{tabular}
&
\begin{tabular}[c]{@{}c@{}}
$f_{\rm Con}$: $mn$
\end{tabular}
\\
\hline
Actuator (A)
&
\begin{tabular}[c]{@{}c@{}}
$\Dec_{\mathrm{HE}}$: $m$
\end{tabular}
&
\begin{tabular}[c]{@{}c@{}}
$f_{\mathrm{Dec}}$: $mn+m$\\
$Sum$: $mn$
\end{tabular}
\\
\hline
\end{tabular}
\end{table}

\begin{proof}
From the equal-accuracy alignment \eqref{eq:eq_align}, it suffices to choose $\delta$ such that
$\rho^{-\delta}\le \varepsilon_q$.
Equivalently, $\delta \ge \log_\rho(1/\varepsilon_q)$, and since $\delta\in\mathbb{N}$ we may take
\[
\delta(\varepsilon_q)=\left\lceil \log_\rho\frac{1}{\varepsilon_q}\right\rceil.
\]
For QE-MPC, the quantizer satisfies $\mathbb{E}[e^2]\le 2^{-2w}$ from \eqref{eq:qec_moments}.
Hence the RMS error is bounded by $\sqrt{\mathbb{E}[e^2]}\le 2^{-w}$, and requiring $2^{-w}\le \varepsilon_q$
is equivalent to $w \ge \log_2(1/\varepsilon_q)$.
With $w\in\mathbb{N}$, a sufficient choice is
\[
w(\varepsilon_q)=\left\lceil \log_2\frac{1}{\varepsilon_q}\right\rceil,
\qquad
p \ge w(\varepsilon_q).
\]
\textbf{Paillier MPC:}
Under the bit-cost model above, $\oplus$, $\odot$, and $\Enc_{\mathrm{HE}}/\Dec_{\mathrm{HE}}$ cost $O(L^2)$, $O(b_KL^2)$, and $O(L^3)$, respectively.
Combining these costs with the primitive counts in Table~\ref{tab:primitive_counts} yields the Paillier column of Table~\ref{tab:bit_bounds_both}.\\
\textbf{QE-MPC:}
Table~\ref{tab:primitive_counts} gives $(n+m)$ calls to $f_{\mathrm{Enc}}$, $mn$ calls to $f_{\mathrm{Con}}$, $(mn+m)$ calls to $f_{\mathrm{Dec}}$, and $O(mn)$ accumulate operations per cycle.
Under the $p$-bit schoolbook model, each $\exp/\log/\mathrm{power}$ evaluation costs $O(p^3)$ and each multiply-add costs $O(p^2)$, yielding the QE-MPC column of Table~\ref{tab:bit_bounds_both}.
\end{proof}

For standard security levels, Paillier requires a large modulus bit length $L$ (e.g., $L=1024/2048$ \cite{alexandru2020cloud}), whereas the quantum-encrypted realization operates at a much smaller finite precision $p$ \cite{ren2025quantum}. Hence, in typical regimes where $p\ll L$, QE-MPC avoids the dominant $O(L^3)$ big-integer costs, requires smaller payload, and achieves lower online complexity. Here, payload denotes the number of transmitted bits per control cycle.

\section{Numerical Example}
\subsection{Simulation Setup}
We consider the battery-current loop of a hybrid battery--ultracapacitor power source, a standard benchmark in power-electronic energy management \cite{6514560}. Following the battery--converter model in \cite{6514560}, we use a disturbance-augmented discrete-time representation obtained by ZOH discretization with $T_s^{\mathrm{batt}}=0.05~\mathrm{s}$, and treat the constant current reference as a parameter, yielding $\theta(k)\in\mathbb{R}^6$. The output is the battery current, the input is the converter modulation index, and the battery state of charge is enforced through bounds on $v_{Cb}$.

The explicit MPC uses horizon $N_{\mathrm{batt}}=5$ and stage cost
\begin{equation}
J=\sum_{k=0}^{N_{\mathrm{batt}}-1}\!\left(
w_{i\mathrm{batt}}\big(i_{\mathrm{batt}}(k)-i_{\mathrm{batt}}^{\mathrm{req}}\big)^2
+r_{m\mathrm{batt}}\,m_{\mathrm{batt}}(k)^2
\right),
\label{eq:cost_batt}
\end{equation}
with $w_{i\mathrm{batt}}=r_{m\mathrm{batt}}=1$. The constraints are $m_{\mathrm{batt}}\in[0.1,0.9]$, $v_{Cb}\in[5.85,6.45]~\mathrm{V}$, and $i_{\mathrm{batt}}(k)+d_{\mathrm{batt}}(k)\in[i_{\mathrm{batt}}^{\mathrm{req}}-1,\ i_{\mathrm{batt}}^{\mathrm{req}}+1]$, with $i_{\mathrm{batt}}^{\mathrm{req}}\in[-2,2]~\mathrm{A}$. The plaintext controller is synthesized offline as a 6D explicit MPC law with 45 regions and used as the common benchmark for Paillier, RSA, AES, and QE-MPC, all with the same PWA partition and gains.

\subsection{Effectiveness}
We validate QE-MPC by implementing Protocol~\ref{prot:qe_empc_protocol} and comparing it with plaintext explicit MPC under the same region selection $\sigma=\pi(x(k))$.
Fig.~\ref{fig:tracking} shows indistinguishable closed-loop responses, and the input mismatch $|u_{\mathrm{QEMPC}}-u_{\mathrm{plain}}|$ remains at numerical roundoff level throughout the simulation.
The average mismatch is $1.70\times10^{-12}\,\mathrm{A}$, confirming exact recovery of the explicit PWA law~\eqref{eq:correctness} up to machine precision, consistent with Theorem~\ref{thm:effectiveness}.

\begin{figure}[htbp]
    \centering
    \includegraphics[width=0.9\linewidth]{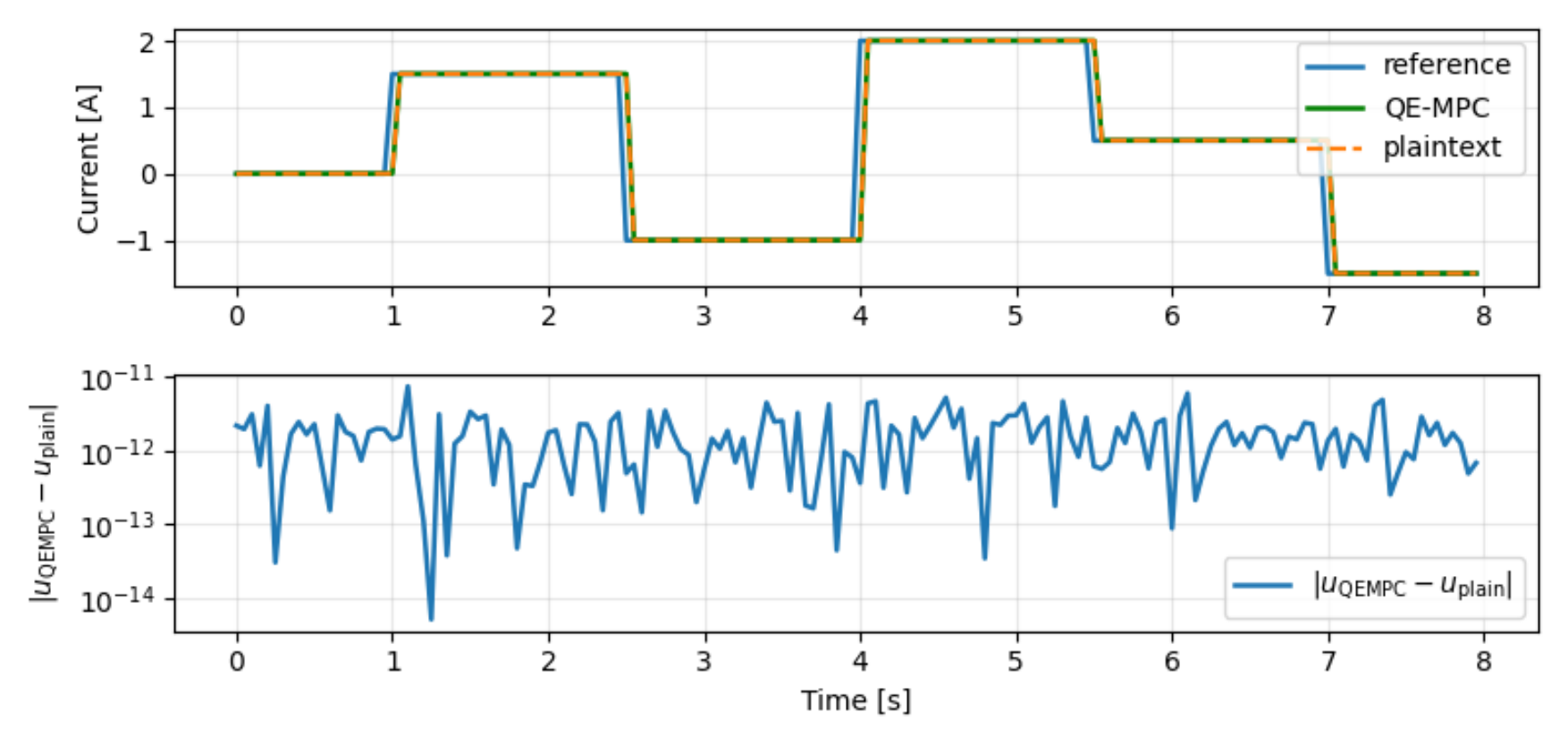}
    \caption{Effectiveness of QE-MPC. Top: closed-loop tracking responses of plaintext explicit MPC and QE-MPC under the same reference signal. Bottom: per-step input mismatch $|u_{\mathrm{QE-MPC}}-u_{\mathrm{plain}}|$, showing machine-precision agreement between encrypted and plaintext evaluations.}
    \label{fig:tracking}
\end{figure}

\begin{figure}[htbp]
    \centering
    \includegraphics[width=0.9\linewidth]{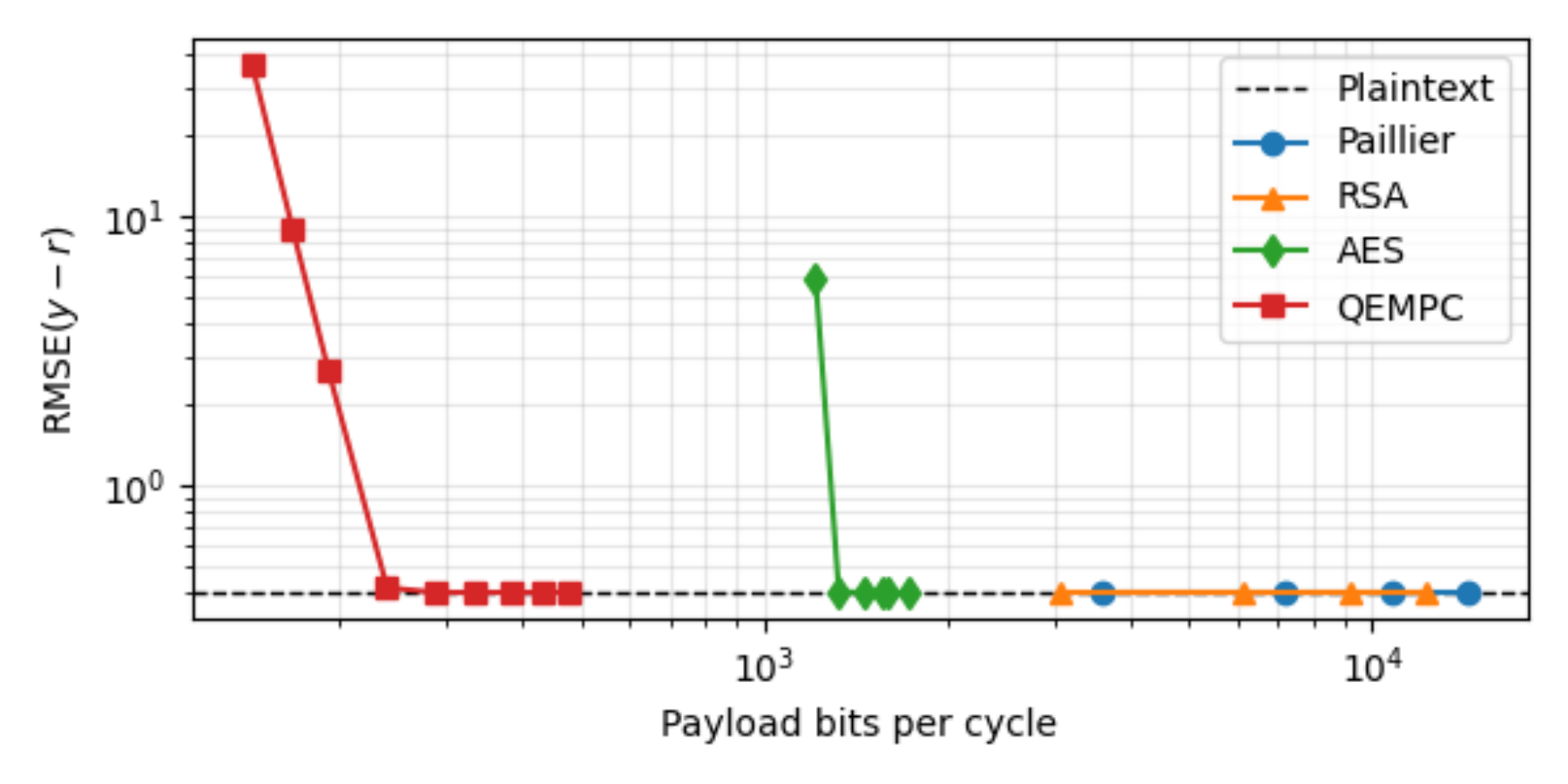}
    \caption{Tracking error $\mathrm{RMSE}(y-r)$ versus payload bits per cycle for plaintext, Paillier, RSA, AES, and QE-MPC. QE-MPC attains the same low-error regime with a substantially smaller payload than the classical encrypted baselines.}
    \label{fig:bits_vs_rmse}
\end{figure}

\subsection{Computational Efficiency}
We compare QE-MPC with implementations based on Paillier, RSA, and AES under matched numerical accuracy $\varepsilon_q$.
Fig.~\ref{fig:bits_vs_rmse} shows that QE-MPC attains the same error regime with a substantially smaller payload than the classical baselines.
Table~\ref{tab:timing_compare} reports the average per-cycle execution time measured with Python on an 3.8\,GHz AMD Ryzen~7~9700X.
QE-MPC is substantially faster than Paillier and RSA because it avoids costly large-integer arithmetic.
Although AES is fastest overall, it requires a much larger payload than QE-MPC to reach the same tracking accuracy.
The results are consistent with Theorem~\ref{thm:bit_ops_equal_accuracy}.

\begin{table}[htbp]
\centering
\caption{Per-cycle execution time comparison among Paillier, RSA, AES, and QE-MPC ($ms$)}
\label{tab:timing_compare}
\begin{tabular}{lcccc}
\hline
Module & Paillier & RSA & AES & QE-MPC \\
\hline
Sensor      & 1.325 & 0.153 & 0.135 & 0.175 \\
Controller  & 0.259 & 0.017 & 0.013 & 0.030 \\
Actuator    & 0.184 & 0.331 & 0.003 & 0.016 \\
\hline
Total cycle & 1.771 & 0.503 & 0.153 & 0.223 \\
\hline
\end{tabular}
\end{table}

\subsection{Confidentiality}
We evaluate confidentiality against a least-squares (LS) identification adversary following \cite{10106403}.
The adversary observes the transmitted proxy trajectory $\tilde x(k)$ and the known input $u(k)$, fits a linear one-step predictor, and rolls it out from the known initial state $x(0)$ to obtain $\hat x(k)$.
Confidentiality is measured by the average relative error $\frac{1}{T}\sum_{k=0}^{T-1}\|\hat x(k)-x(k)\|_2/\|x(k)\|_2$, averaged over 1000 trials under four perturbation settings.
Table~\ref{tab:ls_attack} shows negligible reconstruction error for plaintext, but errors of order $10^{-2}$ for Paillier, RSA, AES, and QE-MPC across all settings.
Hence, the encrypted proxies do not enable accurate recovery of the plaintext trajectory, and QE-MPC achieves confidentiality comparable to the classical encrypted baselines.

\begin{table}[htbp]
\centering
\caption{Average eavesdropping error across encryption schemes under different noises ($\times10^{-2}$)}
\label{tab:ls_attack}
\begin{tabular}{l|ccccc}
\hline
Noise & Plaintext & Paillier & RSA & AES & QE-MPC \\
\hline
None     & 0.00 & 1.34 & 1.17 & 1.05 & 1.38 \\
Gaussian & 0.03 & 1.34 & 1.17 & 1.05 & 1.36 \\
Uniform  & 0.21 & 1.33 & 1.18 & 1.06 & 1.24 \\
Impulse  & 0.59 & 1.30 & 1.18 & 1.09 & 1.20 \\
\hline
\end{tabular}
\end{table}

\section{Conclusions}
We proposed a quantum enhanced architecture for encrypted explicit MPC evaluation that enables lightweight encryption, computation, and decryption of the piecewise affine control law without relying on costly public key homomorphic operations.
We established exact recovery of the plaintext explicit MPC law and showed that, under matched control accuracy, QE-MPC achieves lower online computational complexity than classical encrypted MPC baselines.
The method was validated on the battery current loop of a hybrid battery--ultracapacitor power source, where the encrypted implementation preserved closed loop tracking with negligible numerical mismatch.
Numerical results showed that QE-MPC reduces payload and runtime relative to public key baselines while providing confidentiality against the considered LS attack comparable to other encrypted schemes.
These results support QE-MPC as an efficient approach to secure explicit MPC for networked systems.

\bibliographystyle{IEEEtran}
\bibliography{reference}

\end{document}